\documentclass[11pt]{article}

\usepackage{amsfonts, amsthm}
\usepackage{amssymb}
\usepackage{amsmath}  
\usepackage{graphicx, comment, color, hyperref}
\usepackage{mathrsfs}
\usepackage{authblk}

\usepackage[a4paper]{geometry}
\newtheorem{theorem}{Theorem}

\newtheorem{lemma}[theorem]{Lemma}

\theoremstyle{remark}
\newtheorem{remark}{Remark}

\newcommand{\ket}[1]{|#1\rangle}
\newcommand{\bra}[1]{\langle#1|}

\newcommand{\tr}{\text{\rm tr}}

\newcommand{\Ent}{\text{\rm Ent}}
\newcommand{\rank}{\text{\rm rank}}

\newcommand{\cH}{\mathcal{H}}

\newcommand{\cK}{\mathcal{K}}
\newcommand{\cI}{\mathcal{I}}
\newcommand{\dd}{\text{\rm{d}}}
\newcommand{\bL}{\mathcal{B}}

\newcommand{\cL}{\mathcal{L}}

\newcommand{\bfs}{\mathbf{s}}

\title{Improved Quantum Hypercontractivity Inequality for the Qubit Depolarizing Channel} 

\author{Salman Beigi}
\affil{\it \small School of Mathematics, Institute for Research in Fundamental Sciences (IPM)\\ \it \footnotesize P.O.~Box 19395-5746, Tehran, Iran}

\date{December 9, 2021}

\begin{document}

\maketitle

\begin{abstract}
The hypercontractivity inequality for the qubit depolarizing channel $\Psi_t$ states that $\|\Psi_t^{\otimes n}(X)\|_p\leq \|X\|_q$ provided that $p\geq q> 1$ and $t\geq \ln \sqrt{\frac{p-1}{q-1}}$. In this paper we present an improvement of this inequality. We first prove an improved quantum logarithmic-Sobolev inequality and then use the well-known equivalence of logarithmic-Sobolev inequalities and hypercontractivity inequalities to obtain our main result. As applications of these results, we present an asymptotically tight quantum Faber-Krahn inequality on the hypercube, and a new quantum Schwartz-Zippel lemma. 

\end{abstract}

%******************************************************************************
\section{Introduction}\label{sec:intro}
Hypercontractivity inequalities express a bound on the norm of a ``noisy" version of a function in terms of the norm of the function itself. When the noise is Markovian and the noise operator belongs to a continuous semigroup, hypercontractivity inequalities are proven using the log-Sobolev inequalities. Such inequalities bound the entropy of a function in terms its \emph{energy}, called the Dirichlet form. Hypercontractivity and log-Sobolev inequalities~\cite{G75b, Nelsen, SH-K72, Beckner, Bonami} have found several applications, e.g., in concentration of measure inequalities,~\cite{BLM13, RS13} transportation cost inequalities,~\cite{GL10} estimating the mixing times,~\cite{DSC96} analysis of Boolean functions~\cite{deWolf08} and strong converse bounds in information theory.~\cite{AG76, KA12}

Hypercontractivity and log-Sobolev inequalities can be  generalized to the (non-commutative) quantum case.~\cite{Zegar, KT13} In this case, the noise operator is a quantum superoperator and belongs to a quantum Markov semigroup. As in the classical case, quantum hypercontractivity and log-Sobolev inequalities have found  several applications (see, e.g.,~\cite{KT13, DB14, KT2, CKMT15, BDR20, Bardetetal19, CDR19, BJLRS21}). Nevertheless, due to the non-commutative nature of the quantum theory, proving quantum hypercontractivity and log-Sobolev inequalities is usually more challenging comparing to their classical counterparts.

An important quantum superoperator is the \emph{qubit} depolarizing channel
defined by
$$\Psi_t(X) = e^{-t}X + (1-e^{-t}) \tr(X) \frac{I}{2}, \qquad t\geq 0.$$
Here, the superoperators are parametrized in such a way that they satisfy $\Psi_s\circ \Psi_t= \Psi_{s+t}$ and form a semigroup. 
The hypercontractivity inequalities for this quantum Markov semigroup, first proven in~\cite{MO10, King14}, take the form 
\begin{align}\label{eq:NC-HC}
\|\Psi_t^{\otimes n}(X)\|_p\leq \|X\|_q, \qquad  \text{ if }\qquad e^{2t}\geq \frac{p-1}{q-1}.
\end{align}
Here, it is assumed that $p, q> 1$ and 
$$\|X\|_p = \Big( \frac {1}{2^n} \tr(|X|^p) \Big)^{\frac 1p},$$
where $|X|=\sqrt{X^\dagger X}$. This inequality in the classical (commutative) case is a well-known hypercontractivity inequality which takes the same form, but $X$ is restricted to be diagonal in the computational basis (see, e.g.,~\cite{deWolf08}). Moreover, this inequality (without imposing any assumption on $X$) is tight even if we restrict $X$ to be diagonal in the computational basis.

The proof of inequality~\eqref{eq:NC-HC} for $n=1$ is a simple consequence of its classical counterpart. Nevertheless, unlike the classical case, its proof for arbitrary $n$ is highly non-trivial. 
In the commutative case, having~\eqref{eq:NC-HC} for $n=1$, its proof for arbitrary $n$ is immediate using the multiplicativity of the operator norm, or using a certain subadditivity of the entropy function. In the non-commutative case, the latter two properties do not hold in general, so proving~\eqref{eq:NC-HC} for arbitrary $n$ needs new tools. In~\cite{MO10, King14} this inequality is proven using an inequality on the norms of $2\times 2$ block matrices by King.~\cite{King03}

Our main result in this paper is an improvement of~\eqref{eq:NC-HC}. Our result is inspired by and is a quantum generalization of the works of Polyanskiy and Samorodnitsky.~\cite{Samor, PSamor} Again, our result in the case of $n=1$ is easily derived from~\cite{Samor, PSamor}. Then, to take the induction step we use an entropic inequality of~\cite{BDR20}, that itself is based on~\cite{King03}. 

In the following section after introducing some notations, we formally state our main results in Theorem~\ref{thm:main-LS} and Theorem~\ref{thm:main-HC}. Next, in Section~\ref{sec:application} we present some applications. In particular, we present an asymptotically tight quantum Faber-Krahn inequality on the hypercube. Moreover, recalling the Fourier expansion of operators in the Pauli basis and the form of $\Psi_t^{\otimes n}$ in that basis, we establish an improved quantum Schwartz-Zippel lemma that puts a bound on the ``degree" of an operator in terms of its rank. The proofs of the main theorems, after some auxiliary lemmas in Section~\ref{sec:aux}, come in Sections~\ref{sec:-HC} and~\ref{sec:proof-LS}. We conclude with some final remarks in Section~\ref{sec:final}.

%*********************************************************************************
\section{Main results} 
Throughout the paper, we use $\cH=\mathbb C^2$. We let $\mathcal B(\cH^{\otimes n})$ be the space of linear operators acting $\cH^{\otimes n}$, and for $X\in \cH^{\otimes n}$ let 
$$\tau(X)=\frac{1}{2^n}\tr(X),$$
be the \emph{normalized} trace. The Hilbert-Schmidt inner product on $\bL(\cH^{\otimes n})$ is given by
$$\langle X, Y\rangle := \tau(X^{\dagger} Y).$$

Let $\cL:\mathcal{B}(\cH)\rightarrow \mathcal{B}(\cH)$ be the \emph{Lindblad generator}
$$\cL(X)=X- \tau(X) I.$$
Then, using $\cL^2=\cL$, the depolarizing channel equals
$$\Psi_t(X) = e^{-t\cL}(X) = e^{-t}X + (1-e^{-t}) \tau(X)I.$$
The Lindblad generator associated to the $n$-fold tensor product superoperator $\Psi_t^{\otimes}$ is equal to 
\begin{align}\label{eq:K}
\cK_n = \sum_{i=1}^n \widehat \cL_i,
\end{align}
where
\begin{align}\label{eq:def-hat-L}
\widehat \cL_i:= \cI^{\otimes (i-1)}\otimes \cL\otimes \cI^{\otimes (n-i)}.
\end{align} 
Here $\cI:\mathcal{B}(\cH)\rightarrow \bL(\cH)$ is the identity superoperator. Indeed, we have $e^{-t\cK_n} = \Psi_t^{\otimes n}$.

The \emph{entropy} of a positive semidefinite $X\in \bL(\cH^{\otimes n})$ is given by
$$\Ent(X) = \tau(X\ln X) - \tau(X) \ln \tau(X).$$
The term entropy for this function is used since it is related to the quantum relative entropy; if $X=2^n\rho$, where $\rho$ is a density matrix, then $\Ent(X) =\tr(\rho \ln \rho) +\ln 2^n =D(\rho\| 1/2^n I) $, where  $D(\rho\| 1/2^n I)$ is the quantum relative entropy between $\rho$ and the maximally mixed state.

The Schatten $p$-norm is defined by
$$\|X\|_p := \big(  \tau(|X|^p) \big)^{\frac 1p},$$
where as before $|X|= \sqrt{X^{\dagger}X}$.
The appearance of the entropy function in the study of hypercontractivity inequalities stems from the fact that the derivative of the norm function is given in terms of entropy.

\begin{lemma}\label{lem:entropy-derivative}
Let $X$ be a positive semidefinite matrix. Then the followings hold:
\begin{itemize}
\item[{\rm (i)}] We have
$$\frac{\dd}{\dd p}  \|X\|_p = \frac{1}{p^2 \|X\|_p^{p-1}}\Ent(X^p).$$

\item[{\rm (ii)}]
Let $f(t) = \|\Psi_t^{\otimes n}(X)\|_{p(t)} $. Then we have 
$$f'(t) = \frac{1}{\gamma}\Big(\,\frac{p'(t)}{p(t)^2} \Ent(Y^2) - \big\langle Y^{2/p(t)}, \cK_n Y^{2-2/p(t)}     \big\rangle\,\Big),$$
where $Y=\big( \Psi_{t}^{\otimes n}(X) \big)^{p(t)/2}$ and $\gamma = \big\|\Psi_{t}^{\otimes n}(X)\big\|_{p(t)}^{p(t)-1}$.
\item[{\rm (iii)}] For any $s>1$ we have
$$\frac{\Ent(X^s)}{\tau(X^s)}\geq \frac{\ln \|X\|_s - \ln \|X\|_1}{1-1/s}.$$

\end{itemize}
\end{lemma}

\begin{proof}
Parts (i) and (ii) follow by straightforward computations and are standard (see, e.g., Lemma 12 of~\cite{CKMT15}). For the third part we use H\"older's inequality for Schatten $p$-norm~\cite{MBhatia} saying that 
\begin{align}\label{eq:Holder-ineq}
\|AB\|_r\leq \|A\|_p\, \|B\|_q,
\end{align}
if $p, q, r>0$ and $1/r = 1/p +1/q$. Then for any $x, y\geq 0$ we have
$$\|X\|_{2/(x+y)}^2 = \|X^2\|_{1/(x+y)}\leq \|X\|_{1/x}\cdot \|X\|_{1/y}.$$
This means that the function 
$$g(x)= \ln \|X\|_{\frac 1x},$$
is convex. Thus, for all $0<x<1$ we have
$$g'(x)\leq \frac{g(x)-g(1)}{x-1}.$$
Computing $g'(x)$ at $x=1/s$  using part (i), we obtain the desired inequality. 
\end{proof}

The binary entropy function for $s\in [0,1]$ is 
\begin{align}\label{eq:BEF}
h(s) = -s\ln s - (1-s)\ln(1-s).
\end{align}
We note that $h(s)$ is symmetric about $s=1/2$ and its restriction to $[0,1/2]$ is monotone increasing. Thus, $h^{-1}:[0,\ln 2]\rightarrow [0, 1/2]$ is well-defined.

We can now state a main result of this paper. 

\begin{theorem}\label{thm:main-LS}
For every positive semidefinite $X\in \bL(\cH^{\otimes n})$ we have
\begin{align}\label{eq:main-LSI}
\alpha(\xi)\cdot \Ent(X^2)\leq \langle X, \cK_n X\rangle, 
\end{align}
where 
$$\xi = \frac{\Ent(X^2)}{n \, \tau(X^2)},$$
and $\alpha:[0, \ln 2]\rightarrow [1/2, 1/(2\ln 2)]$ is given by
\begin{align}\label{eq:def-c}
\alpha(\xi)= \frac{1}{\xi}\Big(\frac 12 - \sqrt{h^{-1}(\ln 2-\xi) \cdot(1-h^{-1}(\ln 2-\xi))    }\Big).
\end{align}

\end{theorem}

\begin{figure}
\begin{center}
\includegraphics[height=2.25in]{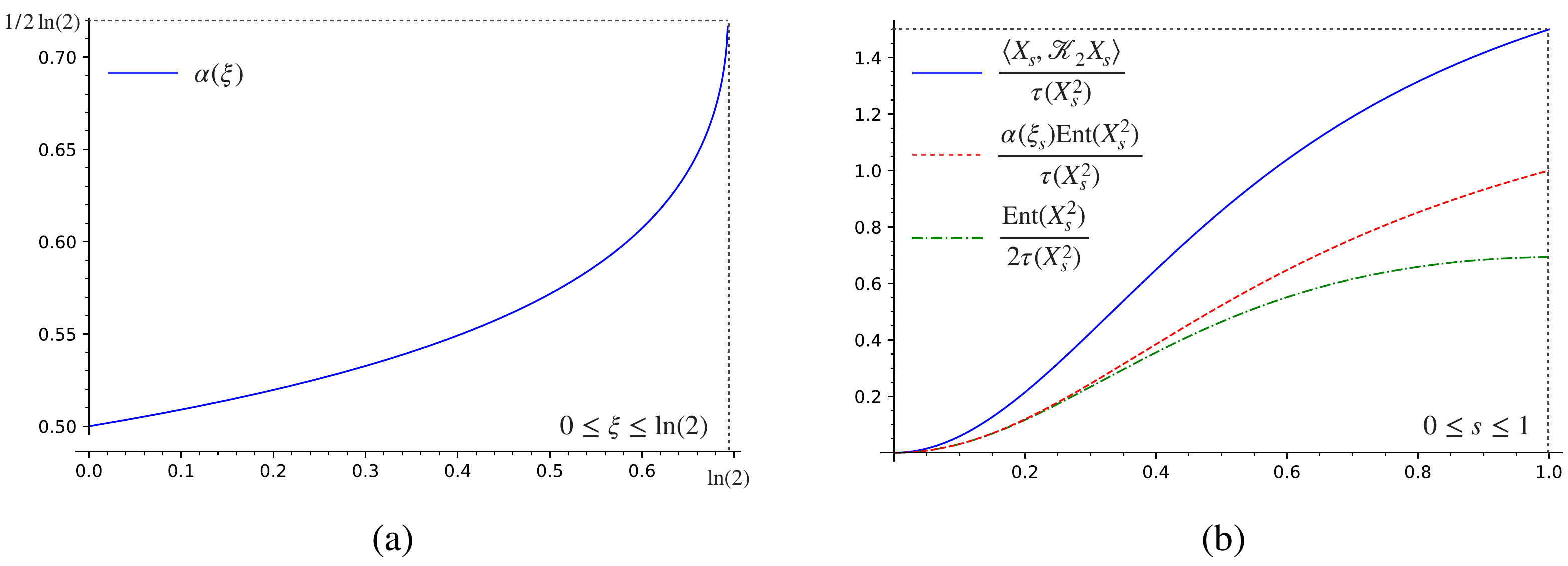}
\caption{\footnotesize (a) The plot of function $\alpha(\xi)$ defined in~\eqref{eq:def-c}. We note that as will be proven in Lemma~\ref{lem:phi-convex} this function is monotone increasing. (b) For $0\leq s\leq 1$  let $X_s= s\ket{\psi}\bra{\psi}+(1-s)\rho\otimes \rho$ where $\rho=I/2$ and $\ket\psi=\frac{1}{\sqrt 2}(\ket{00}+\ket{11})$. By Theorem~\ref{thm:main-LS},  $\alpha(\xi_s)\Ent(X_s^2)$ is a lower bound on $\langle X_s, \mathcal K_2 X_s\rangle$ that is an improvement over $\frac 12 \Ent(X_s^2)$, the standard log-Sobolev inequality. We note that both lower bounds are tight at $s=0$ where $X_s$ is a product state, while for larger values of $s$ the improvement provided by our bound becomes apparent. Here, in the plots we divide all the quantities by $\tau(X_s^2)$ for the sake of normalization. This is the same as replacing $X_s$ with $X_s/\|X_s\|_2$.
}
\label{fig:plots}
\end{center}
\end{figure}

\medskip
\begin{remark}
In the statement of Theorem~\ref{thm:main-LS} we implicitly assume that $\xi$ belongs to the domain of function $\alpha(\cdot)$ which is $[0,\ln 2]$. We will prove this fact in Appendix~\ref{app:rho}.
\end{remark}

\begin{remark}\label{rem:optimal}
We note that $\alpha(\xi)\geq 1/2$ for all $\xi\in [0, \ln 2]$ (see Figure~\ref{fig:plots}(a)). Thus,  inequality~\eqref{eq:main-LSI} is stronger than the  log-Sobolev inequality of~\cite{KT13} for the depolarizing channel. 
\end{remark}

\begin{remark}\label{rem:optimal}
It is not hard to verify that~\eqref{eq:main-LSI} is tight for any $X$ of the form $X=\rho^{\otimes n}$ where $\rho$ is an arbitrary one-qubit positive semidefinite operator.
\end{remark}

Let us examine inequality~\eqref{eq:main-LSI} for the choice of $X_s=\ket{\psi_s}\bra{\psi_s}$ where 
$$\ket{\psi_s} = \sqrt s\ket{00} + \sqrt{1-s}\ket{11}.$$
Since $X_s$ for any $0\leq s\leq 1$ is a rank-one projection, $\tau(X_s^2) =1/4$ and $\Ent(X_s^2) = \frac 12\ln (2)$ are independent of $s$, resulting in $\xi_s=\Ent(X_s^2)/2\tau(X_s^2)=\ln(2)$.  On the other hand, a simple computation verifies that 
$\langle X_s, \mathcal K_2 X_s\rangle= \frac 12 - \frac 14\big(s^2+(1-s)^2\big).$
Thus, inequality~\eqref{eq:main-LSI} gives 
$$\alpha(\xi_s)\Ent(X_s^2) = \frac 14\leq \frac 12 - \frac 14\big(s^2+(1-s)^2\big),$$
that is tight at extreme points $s\in \{0,1\}$. This inequality should be compared to the standard log-Sobolev inequality which gives the lower bound of $\frac 12\Ent(X_s^2)=\frac 14\ln(2)$ on the right hand side that is nowhere tight.  

Another interesting example is when $X$ is a mixture of a pure state and a product state. Such an example is illustrated in Figure~\ref{fig:plots}(b).

\medskip
Using the above improved log-Sobolev inequality, we can establish the following improved quantum hypercontractivity inequality.

\begin{theorem}\label{thm:main-HC}
Let $p_0>1$ and $X\in \bL(\cH^{\otimes n})$ be such that $\|X\|_{p_0}\geq e^{nr_0}\|X\|_1$ for some $ r_0\geq 0$. 
Let $u(t)$ be the solution of the differential equation
\begin{align}
&u'(t) = \alpha\big(r_0(1+ e^{-u(t)}) \big),\label{eq:diff-eq}\\
&u(0) = \ln(p_0-1),\nonumber
\end{align}
where $\alpha(\xi)$ is defined in~\eqref{eq:def-c}.
Then,  we have
\begin{align}\label{eq:main-HC}
\|\Psi_t^{\otimes n}(X)\|_{p(t)}\leq \|X\|_{p_0},
\end{align}
where $p(t) = 1+ e^{u(4t)}$. In particular, for any $t\geq 0$ we have
\begin{align}\label{eq:main-HC-weaker}
\|\Psi_t^{\otimes n}(X)\|_{1+(p_0-1)e^{4\alpha(r_0)t}}\leq \|X\|_{p_0}. 
\end{align}
\end{theorem}

Observe that (e.g., by H\"older's inequality) the assumption $\|X\|_{p_0}\geq e^{nr_0}\|X\|_1$ always holds for $r_0=0$. In this case, $\alpha(r_0)=1/2$, and~\eqref{eq:main-HC-weaker} is equivalent to~\eqref{eq:NC-HC}. Thus,  choosing the optimal $r_0\geq 0$, the above theorem gives improvements of~\eqref{eq:NC-HC}.

\medskip
\begin{remark}
The differential equation~\eqref{eq:diff-eq} makes sense only if $r_0(1+e^{-u(t)})$ belongs to the domain of $\alpha(\cdot)$, meaning that $r_0(1+e^{-u(t)})\leq \ln 2$. 
To prove this fact observe that by H\"older's inequality~\eqref{eq:Holder-ineq} we have
\begin{align*}
\|X\|_{p_0} &\leq \|X^{1/p_0}\|_{p_0} \cdot \|X^{1-1/p_0}\|_\infty\\
& = \|X\|_1^{1/p_0} \cdot \|X\|_\infty^{1-1/p_0}\\
&\leq \|X\|_1^{1/p_0} \cdot \big( 2^n \|X\|_1  \big)^{1-1/p_0}\\
& =  2^{n(1-1/p_0)} \|X\|_1.
\end{align*}
Then comparing to the assumption $\|X\|_{p_0}\geq e^{nr_0}\|X\|_1$, we find that $r_0\leq (1-1/p_0)\ln 2$. On the other hand, since $\alpha(\cdot)$ is positive, $u(t)$ is monotone increasing and we have $u(t)\geq u(0) = \ln(p_0-1)$. Therefore,
$$r_0(1+e^{-u(t)})\leq (1-1/p_0)\ln 2 (1+e^{-\ln(p_0-1)})= \ln 2.$$

\end{remark}

%***********************************************************
\section{Applications}\label{sec:application}

In this section, before getting to the proofs of our main results, Theorem~\ref{thm:main-LS} and Theorem~\ref{thm:main-HC}, we present some applications.

We first note that inequality~\eqref{eq:main-LSI} together with the \emph{quantum Stroock-Varopoulos inequality}~\cite{CKMT15, BDR20} gives other improved log-Sobolev inequalities. In particular, by the quantum Stroock-Varopoulos inequality we have
$$\langle X, \cK_n X\rangle\leq \frac{1}{4}\langle \ln X^2, \cK_n X^2\rangle.$$
Using this in~\eqref{eq:main-LSI} we obtain
$$4\alpha(\xi)\cdot \Ent(X^2)\leq \langle \ln X^2, \cK_n X^2\rangle.
$$ 
This inequality is an improvement over Theorem 19 of~\cite{BDR20} in the qubit case that gives $\Ent(X^2)\leq \langle \ln X^2, \cK_n X^2\rangle$. This is an improvement since $4\alpha(\xi)\geq 2$. The above inequality is important due to its applications in proving \emph{reverse hypercontractivity inequalities}\cite{BDR20} and bounding the \emph{entropy production of quantum channels}.\cite{MHFW16}

\medskip
As a second application we present an asymptotically tight quantum version of the \emph{Faber-Krahn inequality}. 

\begin{theorem}\label{thm:FK}
Let $X\in \bL(\cH^{\otimes n})$ be a positive semidefinite operator with rank $R$.  Then we have
$$\frac 12- \sqrt{h^{-1}\left(\frac{\ln R}{n}\right)\cdot \left(1-h^{-1}\left(\frac{\ln R}{n}\right)\right)}  \leq \frac{\langle X, \cK_n X\rangle}{n\,\tau(X^2)}.$$

\end{theorem}

We note that this bound is already known to be asymptotically tight in the commutative case  in the limit of $n\to \infty$ when $(\ln R)/n$ is fixed.~\cite{FT05, Samor} To this end, let $X$ be a diagonal operator in the computational basis whose support consists of computational basis vectors with bounded hamming weight. See Appendix C of~\cite{FT05} for more details. 

\begin{proof}
By Theorem~\ref{thm:main-LS} we have
$$\varphi(\xi)=\xi\alpha(\xi) \leq \frac{\langle X, \cK_n X\rangle}{n\tau(X^2)},$$
where $\xi=\Ent(X^2)/n\tau(X^2)$.
We note that 
$$\Ent(X^2)\geq \tau(X^2)\ln \frac{\tau(X^2)}{\tau(X)^2}.$$
To verify this inequality, we can assume with no loss of generality that $X$ is diagonal and $\tau(X^2)=1$. Then, using the concavity of the log function we have $\tau(X^2\ln X) = -\tau(X^2\ln X^{-1})\geq -\ln \tau(X)$, that is equivalent to the above inequality. 
Next, let $P$ be the orthogonal projection on the support of $X$. Then, by the Cauchy-Schwarz inequality we have
$$\tau(X) = \langle X, P\rangle \leq \|X\|_2\cdot \|P\|_2 = \tau(X^2)^{1/2}\cdot \Big(\frac{R}{2^n}\Big)^{1/2},$$
and
$$\frac{\tau(X^2)}{\tau(X)^2}\geq \frac{2^n}{R}.$$
Putting these together we obtain
$$\xi = \frac{\Ent(X^2)}{n \tau(X^2)}\geq \frac{1}{n}\ln\frac{2^n}{R} =\ln 2- \frac{1}{n}\ln R.$$
Therefore, using the monotonicity of $\varphi(\cdot)$ proven in Lemma~\ref{lem:phi-convex} below, we have
$$\varphi\Big(\ln 2- \frac{1}{n}\ln R\Big)\leq \frac{\langle X, \cK_n X\rangle}{n\tau(X^2)},$$
that is our desired inequality.

\end{proof}

\medskip
For the next application we need to recall the Fourier expansion of qubit operators.~\cite{MO10} Denote the Pauli matrices by 
\begin{align*}
\sigma_0 =\begin{pmatrix}
1 & 0\\
0 & 1
\end{pmatrix}, \quad 
\sigma_1 =\begin{pmatrix}
0 & 1\\
1 & 0
\end{pmatrix}, \quad
\sigma_2 =\begin{pmatrix}
0 & -i\\
i & 0
\end{pmatrix}, \quad
\sigma_3 =\begin{pmatrix}
1 & 0\\
0 & -1
\end{pmatrix},
\end{align*} 
and for $\bfs\in \{0, 1, 2, 3\}^n$ let
$$\sigma_{\bfs} := \otimes_{j=1}^n \sigma_{\bfs_j}.$$
We note that $\big\{\sigma_{\bfs}\,:\,  \bfs\in \{0, 1, 2, 3\}^n\big\}$ forms an orthonormal basis for $\mathcal B(\cH^{\otimes n})$ with respect to the Hilbert-Schmidt inner product. Then, any $X\in \mathcal B(\cH^{\otimes n})$ can be written as
$$X = \sum_{\bfs}  \hat x_{\bfs} \sigma_{\bfs},$$
where $\hat x_{\bfs} = \langle \sigma_{\bfs}, X\rangle$ are called the \emph{Fourier coefficients} of $X$. 
For $\bfs \in \{0, 1, 2, 3\}^n$ we let 
$$|\bfs| = \{j\,:\,   \bfs_j\neq 0\},$$
and define the \emph{degree} of $X$ by
$$\deg(X) =\max\{|\bfs|\,:\,  \hat x_{\bfs}\neq 0 \}.$$

The significance of the Fourier expansion for us is that the superoperator $\Psi_t^{\otimes n}(\cdot)$ is diagonal in the Pauli basis. Indeed, for $1\leq k\leq 3$ we have $\Psi_t(\sigma_k)=e^{-t}\sigma_k$ and $\Psi_t(\sigma_0)=\sigma_0$. Therefore,
\begin{align}\label{eq:Psi-Pauli}
\Psi_t^{\otimes n}(\sigma_{\bfs}) = e^{-t|\bfs|}\sigma_{\bfs},
\end{align}
and
$$\Psi_t^{\otimes n}(X) = \sum_{\bfs} e^{-t|\bfs|}\hat x_\bfs \sigma_\bfs.$$

Now we can state our third application which is a \emph{quantum Schwartz-Zippel lemma}
stating that low-degree operators have high rank.

\begin{theorem}\label{thm:degree}
Let $0\neq X\in \mathcal B(\cH^{\otimes n})$ be an operator with rank $e^{nh(r_1)}\leq 2^n$ for some $0< r_1\leq 1/2$, where $h(\cdot)$ is the binary entropy function defined in~\eqref{eq:BEF}.  Then we have
\begin{align}\label{eq:degree-rank}
\frac{1}{n}\deg(X) \geq \frac 12 -\sqrt{r_1(1-r_1)}.
\end{align}
\end{theorem} 

A quantum Schwartz-Zippel lemma was first proven in Corollary 52 of~\cite{MO10} where it was conjectured that it can be improved. Here, we show that our result is indeed an improvement. Recall that by Corollary 52 of~\cite{MO10} we have
$$\frac{1}{n}\deg(X)\geq \frac 12(\ln 2 - \frac{1}{n}\ln(\rank X))=\frac 12(\ln 2 - h(r_1)).$$
Theorem~\ref{thm:degree} is an improvement over this bound since we know that $\alpha(\ln 2-h(r_1))\geq 1/2$ and then 
$$ \frac 12 -\sqrt{r_1(1-r_1)}\geq \frac 12(\ln 2 - h(r_1)).$$

\begin{proof}
For $k\leq n$ define the superoperator $\Pi_k:\mathcal B(\cH^{\otimes n})\to \mathcal B(\cH^{\otimes n})$ by
\begin{align*}
\Pi_k (\sigma_{\bfs}) =
\begin{cases}
\sigma_{\bfs}, \quad |\bfs|\leq k\\
0, \quad \text{otherwise.}
\end{cases}
\end{align*}
Then using~\eqref{eq:Psi-Pauli} we have
$$\Pi_k\leq e^{tk}\Psi_t^{\otimes n},$$
meaning that $e^{tk}\Psi_t^{\otimes n}-\Pi_k$ is a positive semidefinite superoperator. Thus using the fact that $\Pi_k$ is a super-projector we have
\begin{align*}
\|\Pi_k (X)\|_2^2 & = \langle X, \Pi_k (X)\rangle\\
& \leq e^{tk} \langle X, \Psi_t^{\otimes n}(X)\rangle\\
&\leq e^{tk} \|X\|_q\cdot \|\Psi_t^{\otimes n}(X)\|_p,
\end{align*}
where in the last line we apply H\"older's inequality for $1\leq q\leq p$ with $1/p+1/q=1$.

For $1\leq q\leq 2$ define $\hat q$ by $\frac{1}{\hat q} =\frac 1q-\frac 12 $. Also, let $P$ be the projection on the support of $X$. 
Then, by H\"older's inequality  we have
\begin{align}\label{eq:rank-norm-arg}
\|X\|_q = \|PX\|_q \leq \|P\|_{\hat q}\cdot\|X\|_2 = \Big(\frac{e^{nh(r_1)}}{2^n}\Big)^{\frac{1}{\hat q}}\|X\|_2 =  e^{-n\big(\ln 2-h(r_1)\big)(\frac 1q-\frac 12)}\|X\|_2,
\end{align}
where in the second equality we use the fact that $P$ is a projection with  rank $e^{nh(r_1)}$.
Using this for $q=1$, we find that $\|X\|_2\geq e^{nr_0}\|X\|_1$ for 
$$r_0=\frac{1}{2}(\ln 2 - h(r_1)).$$
Thus, we may apply Theorem~\ref{thm:main-HC} with $p_0=2$. That is, 
$$\|\Psi_t^{\otimes n}(X)\|_{p(t)}\leq \|X\|_2, \qquad \quad \forall t\geq 0$$
where $p(t)=1+e^{u(4t)}$ with $u(0)=0$ and $u'(t)=\alpha\big(r_0(1+e^{-u(t)})\big)$.  Then, letting $q(t)=p(t)/\big(p(t)-1\big)\leq 2$ be the H\"older conjugate of $p(t)$, we have
\begin{align*}
\|\Pi_k (X)\|_2^2 \leq e^{tk} \|X\|_{q(t)}\cdot \|\Psi_t^{\otimes n}(X)\|_{p(t)}\leq e^{tk} \|X\|_{q(t)}\cdot \|X\|_{2}.
\end{align*}
Next, once again using~\eqref{eq:rank-norm-arg} for $q=q(t)\leq 2$ we obtain
\begin{align*}
\|\Pi_k (X)\|_2^2 \leq e^{tk} e^{-n\big(\ln 2-h(r_1)\big)(\frac{1}{q(t)}-\frac 12)}\|X\|_2^2 = e^{tk} e^{-n\big(\ln 2-h(r_1)\big)(\frac 12-\frac{1}{p(t)})}\|X\|_2^2.
\end{align*}
We note that this inequality holds for any $t\geq 0$. On the other hand, for small values of $t$, using $p'(0) = 4u'(0)e^{u(0)} = 4\alpha(2r_0)=4\alpha(\ln 2-h(r_1))$ we have
$$\frac{1}{p(t)} = \frac{1}{p(0)} -\frac{p'(0)}{p(0)^2} t + O(t^2)= \frac{1}{2} -\alpha(\ln 2-h(r_1)) t + O(t^2),$$
and then
\begin{align*}
\|\Pi_k (X)\|_2^2 \leq e^{tk} e^{-n\big(\ln 2-h(r_1)\big)(\alpha(\ln 2 -h(r_1))t+O(t^2))}\|X\|_2^2.
\end{align*}

Let $k=\deg(X)$ in which case $\Pi_k (X)=X$.
Then, since $X\neq 0$, the above inequality says that for all sufficiently small $t>0$ we have
$$t\deg(X) -n\big(\ln 2-h(r_1)\big)(\alpha(\ln 2 -h(r_1))t+O(t^2))\geq 0.$$
Equivalently, this means that
$$\frac 1n\deg(X) \geq \big(\ln 2-h(r_1)\big)\alpha(\ln 2 -h(r_1)= \frac 12-\sqrt{r_1(1-r_1)}.$$

\end{proof}

%***********************************************************
\section{Auxiliary lemmas}\label{sec:aux}

As mentioned above, Theorem~\ref{thm:main-LS} and Theorem~\ref{thm:main-HC} for $n=1$ are essentially proven in~\cite{Samor, PSamor}. Here, to take the induction step and prove these results for arbitrary $n$, we need to bound the entropy of an operator in $\mathcal B(\cH^{\otimes n})$ in terms of the entropy of certain operators in $\mathcal B(\cH^{\otimes (n-1)})$.

\begin{lemma}\label{lem:main-ent} \emph{(\cite{BDR20})}
Suppose that $X$ is a positive semidefinite $2N\times 2N$ matrix of the block form
\begin{align}\label{eq:X}
X= \begin{pmatrix}
A & C\\
C^\dagger & B
\end{pmatrix},
\end{align} 
where $A, B, C$ are $N\times N$ matrices. 
Define 
\begin{align}\label{eq:M}
M=\begin{pmatrix}
\|A\|_2 & \|C\|_2\\
\|C^\dagger\|_2 & \|B\|_2
\end{pmatrix}.
\end{align}
Then we have
\begin{align}\label{eq:X^2-M^2-0}
\Ent(X^2)\leq \Ent(M^2) + \frac{1}{2}\Ent(A^2)+\frac{1}{2}\Ent(B^2) +\Ent(|C|^2).
\end{align}
\end{lemma}

This lemma in a more general form is proven in~\cite{BDR20}. Here, we briefly explain the proof idea and refer to~\cite{BDR20} for a detailed proof.

\begin{proof}[Proof sketch.]

For every $p\geq 1$ define 
$$M_p=\begin{pmatrix}
\|A\|_p & \|C\|_p\\
\|C^\dagger\|_p & \|B\|_p
\end{pmatrix}.$$
Thus $M=M_2$. It is shown in~\cite{King03} that for every $p\geq 2$ we have
$$\|X\|_p\leq \|M_p\|_p.$$
That is, letting $g(p):= \|M_p\|_p^p - \|X\|_p^p$, we have $g(p)\geq 0$ for all $p\geq 2$. On the other hand, we note that $g(2)=0$. Thus, we must have $g'(2)\geq 0$. Computing the derivative of $g(p)$ using Lemma~\ref{lem:entropy-derivative}, we find that $g'(2)\geq 0$ is equivalent to~\eqref{eq:X^2-M^2-0}. 

\end{proof}

We need yet another lemma in the proof of Theorem~\ref{thm:main-LS}.

\begin{lemma}\label{lem:2} \emph{(\cite{BDR20})}
For every $C\in \bL(\cH^{\otimes n})$ we have
$$\langle C, \cK_n C \rangle +\langle C^\dagger, \cK_n C^\dagger \rangle \geq \langle |C|, \cK_n |C| \rangle +\langle |C^\dagger|, \cK_n |C^\dagger| \rangle.$$
\end{lemma}

\begin{proof}
Observe that 
$$Z_{\pm} = \begin{pmatrix}
|C^\dagger| & \pm C\\
\pm C^\dagger & |C|
\end{pmatrix}\geq 0,$$
is positive semidefinite.~\cite{PBhatia} On the other hand, $\Psi_t^{\otimes n}$ is completely positive. Therefore, 
$$\begin{pmatrix}
\Psi_t^{\otimes n}(|C^\dagger|) & \pm \Psi_t^{\otimes n}(C)\\
\pm \Psi_t^{\otimes n}(C^\dagger) & \Psi_t^{\otimes n}(|C|)
\end{pmatrix}\geq 0.$$
As a result, 
$$\Bigg\langle \begin{pmatrix}
|C^\dagger| & - C\\
- C^\dagger & |C|
\end{pmatrix},  \begin{pmatrix}
\Psi_t^{\otimes n}(|C^\dagger|) &  \Psi_t^{\otimes n}(C)\\
 \Psi_t^{\otimes n}(C^\dagger) & \Psi_t^{\otimes n}(|C|)
\end{pmatrix}\Bigg\rangle\geq 0.$$
This means that
$$\langle |C|, \Psi_t^{\otimes n}(|C|)\rangle +\langle |C^{\dagger}|, \Psi_t^{\otimes n}(|C^\dagger|)\rangle \geq \langle C, \Psi_t^{\otimes n}(C)\rangle+\langle C^\dagger, \Psi_t^{\otimes n}(C^\dagger)\rangle.$$
This inequality holds for all $t\geq 0$, and turns into an equality for $t=0$. Then, the desired inequality follows once we note that 
$$\frac{\dd}{\dd t}\Psi_t^{\otimes n}\Big|_{t=0} = -\cK_n.$$ 

\end{proof}

Our final lemma presents some properties of the function $\alpha(\xi)$ defined in the statement of Theorem~\ref{thm:main-LS}.

\begin{lemma}\label{lem:phi-convex} \emph{(Theorem 6 of~\cite{PSamor} and Lemma 2.1 of~\cite{Samor})}
The function $\varphi(\xi)$ defined on the interval $[0, \ln 2]$ by
$$\varphi(\xi)=\frac 12 - \sqrt{h^{-1}(\ln 2-\xi) \cdot(1-h^{-1}(\ln 2-\xi))    },$$
is convex and monotone increasing. Moreover, $\alpha(\xi) = \varphi(\xi)/\xi$ is monotone increasing. 
\end{lemma}

\begin{proof}
Let $x\in [0,1/2]$ be such that $\xi=\ln 2- h(x)$. Then, we have
$\varphi(\xi) = \frac{1}{2} - \sqrt{x(1-x)}$. Taking the derivative of both sides with respect to $x$ we obtain
$$\ln\big(\frac{1-x}{x}\big) \varphi'(\xi) = \frac{1-2x}{2\sqrt{x(1-x)}}.$$
Thus, $\varphi'(\xi)\geq 0$ and $\varphi$ is monotone increasing. 

Let $y=(1-x)/x \geq 1$. Then, the above equation can be rewritten as 
$$\varphi'(\xi) = \frac{1}{2\ln y} f(y),$$
where $f(y) = \sqrt y - 1/\sqrt y$. We note that, $\frac{\dd}{\dd x} y = -1/x^2$. Therefore, taking the derivative of both sides of the above equation with respect to $x$ we obtain
$$-\ln(y)\varphi''(\xi) = -\frac{1}{2x^2} \frac{\dd}{\dd y}\Big(\frac{1}{\ln y}f(y)\Big).$$
Hence, to show that $\varphi(\xi)$ is convex, it suffices to show that 
$$\frac{\dd}{\dd y}\Big(\frac{1}{\ln y}f(y)\Big)\geq 0,$$
or equivalently
$$y\ln(y) f'(y)\geq f(y).$$
Define $g(s) = f(e^{2s}) = e^s- e^{-s}$ for $s\geq 0$. Then, using $g'(s)=2e^{2s}f'(e^{2s})$  the above inequality is equivalent to 
$$s g'(s)\geq g(s).$$
This inequality holds since $g(s) =  e^s- e^{-s}$ is a monotone increasing convex function for $s\geq 0$  and $g(0)=0$. 

To verify that $\alpha(\xi)$ is monotone increasing, note that $\alpha'(\xi) = (\xi \varphi'(\xi) - \varphi(\xi))/\xi^2$ that is non-negative since $\varphi(\xi)$ is convex, monotone increasing and $\varphi(0)=0$.

\end{proof}

%************************************************************
\section{Proof of Theorem~\ref{thm:main-LS}}\label{sec:proof-LS}
As argued in the proof of Lemma~\ref{lem:rho} in Appendix~\ref{app:rho}, for $n=1$ we may assume that $X^2$ is diagonal with eigenvalues $1\pm r$ for some $r\in [0,1]$. Then $\tau(X^2)=1$ and $\xi=\Ent(X^2)= \ln 2-h((1-r)/2)$. Therefore, 
$$\alpha(\xi) = \frac{1}{\xi}\Big(\frac{1}{2} - \frac 12 \sqrt{1-r^2}\Big),$$
and 
$$\alpha(\xi)\cdot \Ent(X^2) = \frac{1}{2} \big(1-\sqrt{1-r^2}\big).$$
On the other hand, by a simple computation we have
$$\langle X, \cK_1 X\rangle = \tau(X^2) - \tau(X)^2 = \frac{1}{2} \big(1-\sqrt{1-r^2}\big).$$
Therefore, the desired inequality for $n=1$ holds as an equality.

Now suppose that $n>1$ and that $X\in \mathcal B(\cH^{\otimes n})$ has the block form~\eqref{eq:X}. Let 
$$\varphi(\xi)= \xi \alpha (\xi)=\frac 12 - \sqrt{h^{-1}(\ln 2-\xi) \cdot(1-h^{-1}(\ln 2-\xi))    }.$$
We need to show that 
$$n\tau(X^2)\, \varphi\Big(\frac{\Ent(X^2)}{ n\tau(X^2)}\Big)\leq \langle X, \cK_n X\rangle.$$
By Lemma~\ref{lem:main-ent} we have $\Ent(X^2)\leq e_1+e_2+e_3+e_4$ where
$$e_1= \Ent(M^2),\quad e_2=\frac 1 2 \Ent(A^2),\quad e_3= \frac 1 2 \Ent(B^2),\quad e_4 =\Ent(|C|^2).$$
On the other hand, a simple computation shows that
\begin{align*}
\tau(X^2) = \tau(M^2) = \frac 12\tau(A^2) + \frac 12\tau(B^2) + \tau(|C|^2),
\end{align*}
and 
\begin{align*}
n\tau(X^2) = \tau(M^2) + \frac{n-1}{2}\tau(A^2) + \frac{n-1}{2}\tau(B^2) + (n-1)\tau(|C|^2).
\end{align*}
Therefore, we have
$$\frac{\Ent(X^2)}{n \tau(X^2)} \leq \frac{e_1+e_2+e_3+e_4}{\theta_1+\theta_2+\theta_3+\theta_4},$$ 
where
\begin{align*}
\theta_1=\tau(M^2),\quad \theta_2= \frac{n-1}{2}\tau(A^2),\quad \theta_3= \frac{n-1}{2}\tau(B^2),\quad \theta_4= (n-1) \tau(|C|^2).
\end{align*}
Using the monotonicity and convexity of $\varphi(\cdot)$ proven in Lemma~\ref{lem:phi-convex} we have
\begin{align*}
\varphi\Big(\frac{\Ent(X^2)}{ n \tau(X^2)}\Big) &\leq \varphi\Big(\frac{e_1+e_2+e_3+e_4}{\theta_1+\theta_2+\theta_3+\theta_4}\Big)\\
&\leq \sum_{i=1}^4 \frac{\theta_i}{\theta_1+\theta_2+\theta_3+\theta_4} \,\varphi\big(\frac{e_i}{\theta_i}\big).
\end{align*}
This means that
\begin{align*}
n\tau(X^2) \varphi\Big(\frac{\Ent(X^2)}{ n \tau(X^2)}\Big)& \leq \tau(M^2)  \varphi\Big(\frac{\Ent(M^2)}{\tau(X^2)}\Big) + \frac{n-1}{2}\tau(A^2)\,\varphi\Big(\frac{\Ent(A^2)} {(n-1)\tau(A^2)}\Big)\\
& \qquad +\frac{n-1}{2}\tau(B^2)\,\varphi\Big(\frac{\Ent(B^2)}{(n-1)\tau(B^2)}\Big)\\
&\qquad + (n-1)\tau(|C|^2)\,\varphi\Big(\frac{\Ent(|C|^2)}{(n-1)\tau(|C|^2)}\Big).
\end{align*}
Next, by the induction hypothesis 
and the fact that $\Ent(|C|^2) = \Ent(|C^\dagger|^2)$ and $\tau(|C|^2) = \tau(|C^\dagger|^2)$, we obtain
\begin{align}\label{eq:ptt}
n\tau(X^2) \varphi\Big(\frac{\Ent(X^2)}{ n \tau(X^2)}\Big) & \leq  \langle M, \cL M\rangle + \frac 12 \langle A, \cK_{n-1} A\rangle +\frac 12 \langle B, \cK_{n-1} B\rangle \nonumber\\
&\qquad + \frac 12 \langle |C|, \cK_{n-1} |C|\rangle +\frac 12\langle |C^\dagger|, \cK_{n-1} |C^\dagger|\rangle.
\end{align}
We now note that
\begin{align*}
\langle X, \cK_n X\rangle = \langle X, \cL\otimes \cI^{\otimes (n-1)}(X)\rangle + \Bigg\langle \begin{pmatrix}
A &  C\\
 C^\dagger & B
\end{pmatrix},  \begin{pmatrix}
\cK_{n-1}(A) &  \cK_{n-1}(C)\\
 \cK_{n-1}(C^\dagger) & \cK_{n-1}(B)
\end{pmatrix}\Bigg\rangle.
\end{align*}
We compute each term in the above sum separately: 
\begin{align*}
\langle X, \cL\otimes \cI^{\otimes (n-1)}(X)\rangle  & = \Bigg\langle \begin{pmatrix}
A &  C\\
 C^\dagger & B
\end{pmatrix},  \begin{pmatrix}
\frac{A-B}{2} &  C\\
 C^\dagger & \frac{B-A}{2}
\end{pmatrix}\Bigg\rangle \\
& = \frac{1}{2}\Big(  \frac{1}{2}\tau(A^2) +\frac 12 \tau(B^2) - \tau(AB) + 2\tau(C^\dagger C) \Big)\\
& \geq \frac{1}{2}\Big(  \frac{1}{2}\|A\|_2^2 +\frac 12 \|B\|_2^2 - \|A\|_2\cdot \|B\|_2 + 2\|C\|_2^2 \Big)\\
& = \langle M, \cL M\rangle. 
\end{align*}
For the second term we have
\begin{align*}
\Bigg\langle \begin{pmatrix}
A &  C\\
 C^\dagger & B
\end{pmatrix},  \begin{pmatrix}
\cK_{n-1}(A) &  \cK_{n-1}(C)\\
 \cK_{n-1}(C^\dagger) & \cK_{n-1}(B)
\end{pmatrix}\Bigg\rangle & = \frac 12 \langle A, \cK_{n-1} A\rangle + \frac 12 \langle B, \cK_{n-1} B\rangle \\
&\qquad + \frac 12 \langle C, \cK_{n-1} C\rangle + \frac 12 \langle C^\dagger, \cK_{n-1} C^\dagger\rangle\\
& \geq \frac 12 \langle A, \cK_{n-1} A\rangle + \frac 12 \langle B, \cK_{n-1} B\rangle \\
&\qquad + \frac 12 \langle |C|, \cK_{n-1} |C|\rangle + \frac 12 \langle |C^\dagger|, \cK_{n-1} |C^\dagger|\rangle,
\end{align*}
where the inequality follows from Lemma~\ref{lem:2}.
Putting these in~\eqref{eq:ptt} the desired inequality follows.

%*************************************************************
\section{Proof of Theorem~\ref{thm:main-HC}}\label{sec:-HC}
Let us first assume that $X$ is positive semidefinite. Later we will show how to generalize the argument for arbitrary $X$. 
Define
$$f(t) = \|\Psi_t^{\otimes n}(X)\|_{p(t)} - \|X\|_{p_0}.$$
We need to show that $f(t)\leq 0$ for all $t\geq 0$. 
We note that $f(0)=0$. Suppose that for some $t_1>0$ we have $f(t_1)>0$. Let 
$$s_0 = \sup\{s\,|\, 0\leq s< t_1, \, f(s)=0\}.$$
By the continuity of $f(t)$ we have $f(s_0)=0$. Moreover, $f$ is positive in the interval $(s_0, t_1)$. Then, by the mean value theorem there exits $t_0\in (s_0, t_1)$ such that 
\begin{align}\label{eq:f-t0-positive}
f(t_0)\geq 0,
\end{align}
and 
$$f'(t_0) = f(t_1)/(t_1-s_0)>0.$$
By Lemma~\ref{lem:entropy-derivative} we have
$$\big\|\Psi_{t_0}^{\otimes n}\big\|_{p(t_0)}^{p(t_0)-1}\cdot f'(t_0) = \frac{p'(t_0)}{p(t_0)^2} \Ent(Y^2) - \big\langle Y^{2/p(t_0)}, \cK_n Y^{2-2/p(t_0)}     \big\rangle,$$
where $Y=\big( \Psi_{t_0}^{\otimes n}(X) \big)^{p(t_0)/2}$. Thus $f'(t_0)> 0$ implies
$$\frac{p'(t_0)}{p(t_0)^2} \Ent(Y^2) > \big\langle Y^{2/p(t_0)}, \cK_n Y^{2-2/p(t_0)}     \big\rangle.$$
Next, using the quantum Stroock-Varopoulos inequality~\cite{CKMT15, BDR20} stating that
$$\big\langle Y^{2/p(t_0)}, \cK_n Y^{2-2/p(t_0)}     \big\rangle\geq \frac{4(p(t_0)-1)}{p(t_0)^2}\langle Y, \cK_n Y\rangle, $$
we conclude that
$$\frac{p'(t_0)}{4(p(t_0)-1)} \Ent(Y^2) > \langle Y, \cK_n Y\rangle.$$
On the other hand, Theorem~\ref{thm:main-LS} gives
$$\alpha(\xi)\cdot\Ent(Y^2)\leq \langle Y, \cK_n Y\rangle, $$
where $\xi=\Ent(Y^2)/\big(n\tau(Y^2)\big)$.
Then comparing the above two inequalities, we find that
$$\alpha(\xi) <\frac {p'(t_0)}{ 4(p(t_0)-1)}.$$
Next, using $p'(t_0) = 4 u'(t_0) e^{u(4t_0)}$ and the given differential equation for $u(t)$ we have
$$\frac{p'(t_0)}{ 4(p(t_0)-1)} = \alpha\Big( \frac{r_0 p(t_0)}{p(t_0)-1}  \Big).$$
Comparing the above two inequalities and the monotonicity of $\alpha(\cdot)$ stated in Lemma~\ref{lem:phi-convex}, we obtain
\begin{align}\label{eq:8}
\xi<  \frac{r_0 p(t_0)}{p(t_0)-1}.
\end{align}

Now, part (iii) of Lemma~\ref{lem:entropy-derivative} yields
$$\frac{\Ent(Y^2)}{\tau(Y^2)} = \frac{\Ent\big(\Psi_{t_0}^{\otimes n}(X)^{p(t_0)}  \big)}{\tau\big(\Psi_{t_0}^{\otimes n}(X)^{p(t_0)}\big)}\geq  \frac{\ln \|\Psi_{t_0}^{\otimes n}(X)\|_{p(t_0)} - \ln \|\Psi_{t_0}^{\otimes n}(X)\|_1}{1-1/p(t_0)}.$$
On the other hand, by~\eqref{eq:f-t0-positive} we have $\|\Psi_{t_0}^{\otimes n}(X)\|_{p(t_0)}\geq \|X\|_{p_0}$. Moreover, $X$ and  $\Psi_{t_0}^{\otimes n}(X)$ are positive semidefinite. Hence,  
$$\|\Psi_{t_0}^{\otimes n}(X)\|_1= \tau\big(\Psi_{t_0}^{\otimes n}(X)\big) = \tau(X)=\|X\|_1.$$
Putting these together we conclude that 
$$\frac{\Ent(Y^2)}{\tau(Y^2)}\geq \frac{\ln \|X\|_{p_0} - \ln \|X\|_1}{1- {1}/{p(t_0)}}\geq \frac{nr_0}{1- 1/p(t_0)},$$
which means 
$$\xi>\frac{r_0p(t_0)}{p(t_0)-1}.$$
This is in contradiction with~\eqref{eq:8}.

Now assume that $X$ is arbitrary. Let $Z=|X|=\sqrt{X^\dagger X}$ and $Z'=|X^\dagger|=\sqrt{X X^\dagger}$. Then since $\Psi_t(\cdot)$ is completely positive, we have 
$\|\Psi_t(X)\|^2_{q}\leq \|\Psi_t(Z)\|_q\cdot \|\Psi_t(Z')\|_q$ for any $q\geq 1$. This fact in the special case of hermitian $X$ was first showed in~\cite{AH02}, and its general form can be inferred from the proof of Lemma 2 of~\cite{Watrous05}. Here we present another proof. As already shown in the proof of Lemma~\ref{lem:2} we have
$$\begin{pmatrix}
\Psi_t^{\otimes n}(Z')  &  \Psi_t^{\otimes n}(X)\\
 \Psi_t^{\otimes n}(X^\dagger) & \Psi_t^{\otimes n}(Z)
\end{pmatrix}
=\begin{pmatrix}
\Psi_t^{\otimes n}(|X^\dagger|) &  \Psi_t^{\otimes n}(X)\\
 \Psi_t^{\otimes n}(X^\dagger) & \Psi_t^{\otimes n}(|X|)
\end{pmatrix}\geq 0.$$
This means that there is a \emph{contraction} $K$ such that (see Proposition 1.3.2 of~\cite{PBhatia})
$$\Psi_t^{\otimes n}(X) = \sqrt{\Psi_t^{\otimes n}(Z') }\, K\,\sqrt{\Psi_t^{\otimes n}(Z)}.
$$
Therefore, by H\"older's inequality 
\begin{align*}
\|\Psi_t^{\otimes n}(X)\|_q &= \Big\|\sqrt{\Psi_t^{\otimes n}(Z') } K\sqrt{\Psi_t^{\otimes n}(Z)}\Big\|_q\leq \Big\|\sqrt{\Psi_t^{\otimes n}(Z') }\Big\|_{2q}\cdot \|K\|_\infty \cdot\Big\|\sqrt{\Psi_t^{\otimes n}(Z)}\Big\|_{2q}\\
& \leq \Big\|\sqrt{\Psi_t^{\otimes n}(Z') }\Big\|_{2q}\cdot \Big\|\sqrt{\Psi_t^{\otimes n}(Z)}\Big\|_{2q}\\
& = \Big\|\Psi_t^{\otimes n}(Z') \Big\|_{q}^{1/2}\cdot \Big\|\Psi_t^{\otimes n}(Z)\Big\|_{q}^{1/2}.
\end{align*}
Next, using the fact that $\|Z\|_q=\|X\|_q=\|X^\dagger\|_q=\|Z'\|_q$ we find that
 by assumption $\|Z\|_{p_0}=\|Z'\|_{p_0} =\|X\|_{p_0}\geq e^{nr_0}\|X\|_1 = e^{nr_0}\|Z\|_1=e^{nr_0}\|Z'\|_1$. Then, as we showed that the theorem holds in the case of positive semidefinite operators $Z, Z'$, we have 
\begin{align*}
\|\Psi_t^{\otimes n}(X)\|_{p(t)}&\leq \|\Psi_t^{\otimes n}(Z)\|_{p(t)}^{1/2}\cdot \|\Psi_t^{\otimes n}(Z')\|_{p(t)}^{1/2}\\
& \leq \|Z\|_{p_0}^{1/2}\cdot \|Z'\|_{p_0}^{1/2}\\
&=\|X\|_{p_0}.
\end{align*}

To prove~\eqref{eq:main-HC-weaker}, note that by Lemma~\ref{lem:phi-convex} the function $\alpha(\xi)$ is monotone non-decreasing. Therefore, $u'(t)\geq \alpha(r_0)$ for all $t$, and $u(t)\geq \alpha(r_0) + \ln(p_0-1)$. Using this in~\eqref{eq:main-HC}, and using the monotonicity of $q\mapsto \|X\|_q$, which can be proven using H\"older's inequality, we obtain~\eqref{eq:main-HC-weaker}.

%*******************************************************************************************
\section{Final remarks}\label{sec:final}

One may wonder whether the above results can be extended to the \emph{generalized depolarizing semigroup} with the Lindblad generator $\cL'(X) = X- \tr(\eta X)I$, where $\eta$ is an arbitrary full-rank qubit state (see, e.g.,~\cite{BDR20}).  We note that both Lemma~\ref{lem:main-ent} and Lemma~\ref{lem:2}, as our main tools in proving our results, already hold for arbitrary such $\eta$'s. Thus, it seems that the induction step in the proof of Theorem~\ref{thm:main-LS} can be reproduced for arbitrary $\eta$. Nevertheless, a main issue in such a generalization is the definitions of functions $\alpha(\cdot)$ and $\varphi(\cdot)$ when $\eta$ is not the maximally mixed state. We note that, the functions $\alpha(\cdot)$ and $\varphi(\cdot)$ in the case of $\eta=I/2$ are defined in such a way that for $n=1$ we get equality in~\eqref{eq:main-LSI}. Moreover, we saw that luckily $\varphi(\cdot)$ defined in this way is \emph{convex}, and this convexity is used to execute the induction step. We may take the same path and define $\varphi(\cdot)$ for arbitrary $\eta$ by assuming equality  in~\eqref{eq:main-LSI} for $n=1$. However, even restricting to the completely classical setting (diagonal operators), such a function $\varphi(\cdot)$ would not be convex in general. An approach to overcome this obstacle is to simply replace such a $\varphi(\cdot)$ with its \emph{lower convex envelop}, i.e., the largest convex function that is point-wise smaller. We leave exploring such a generalization for future works.  

The improved log-Sobolev inequality of Theorem~\ref{thm:main-LS} is indeed a log-Sobolev inequality for the parameter $p=2$ (see, e.g.,~\cite{BDR20} for more details). It would be  interesting to extend this result to other values of $p$, as done in~\cite{PSamor} for the classical case. In particular, if we could optimally generalize Theorem~\ref{thm:main-LS} to the case of $p=1$ (an improved quantum modified log-Sobolev inequality), we would be able to prove a quantum version of \emph{Mrs.\ Gerber's Lemma} through the approach of~\cite{PSamor}. As Mrs.\ Gerber's Lemma has several applications in information theory, a quantum version would also be desirable and would find applications. 
We note that as done in Section~\ref{sec:application}, the quantum Stroock-Varopoulos inequality and Theorem~\ref{thm:main-LS} provide us with an improved modified log-Sobolev inequality, but not an optimal one as in the case of $p=2$ (see Remark~\ref{rem:optimal}).  To obtain such an optimal inequality we should take a direct approach. To this end, a main difficulty is to replace the inequality of Lemma~\ref{lem:main-ent} to something relevant to the case of $p=1$.

%%%%%%%%%%%%%%%%%%%%%%%%%%%%%%%%
\appendix

\section{Bounding the entropy in terms of $2$-norm}\label{app:rho}

\begin{lemma}\label{lem:rho}
For every positive semidefinite $X\in \bL(\cH^{\otimes n})$ we have
$$\xi= \frac{\Ent(X^2)}{n\, \tau(X^2)}\in [0, \ln 2].$$
\end{lemma}

\begin{proof}
We need to show that
$$\Ent(X^2)\leq n \ln(2) \tau(X^2).$$
We prove this by induction on $n$.  First, let $n=1$. In this case, both $\Ent(X^2)$ and $\tau(X^2)$ depend only on the eigenvalues of $X^2$. Next, 
by a scaling argument we can assume with no loss of generality that the eigenvalues of $X^2$ are $1\pm r$ for some $r\in[0,1]$. Then $\tau(X^2)=1$ and we have
\begin{align*}
\Ent(X^2) &= \frac{1}{2}\big( (1+r)\ln(1+r) +(1-r)\ln(1-r)  \big) \\
&\leq \frac 12 (1+r)\ln(1+r)\\
&\leq \ln 2.
\end{align*}

Now suppose that $n>1$, and that $X$ is of the block form~\eqref{eq:X} with $A, B, C\in \bL\big(\cH^{\otimes (n-1)}\big)$.
By Lemma~\ref{lem:main-ent} we have
$$\Ent(X^2)\leq \Ent(M^2) + \frac 1 2 \Ent(A^2) +\frac 1 2 \Ent(B^2) + \Ent(|C|^2).$$
On the other hand, a simple computation shows that
\begin{align*}
\tau(X^2) = \tau(M^2) = \frac 12\tau(A^2) + \frac 12\tau(B^2) + \tau(|C|^2),
\end{align*}
and 
\begin{align*}
n\tau(X^2) = \tau(M^2) + \frac{n-1}{2}\tau(A^2) + \frac{n-1}{2}\tau(B^2) + (n-1)\tau(|C|^2).
\end{align*}
Using the induction hypothesis for $M, A, B, |C|$, and the above equations, the desired inequality follows.

\end{proof}

%*******************************************************************************************

{\footnotesize
{~~}
}

\end{document}